\documentclass[runningheads, envcountsame, a4paper]{llncs}
\usepackage{times}

\usepackage[small,it]{caption}
\usepackage[pdftex]{graphicx}
\usepackage{epstopdf} 
\usepackage{microtype}
\usepackage{algorithm2e}
\usepackage{amsmath}
\usepackage{tikz} 
\tikzset{
  treenode/.style = {align=center, inner sep=0pt, text centered,
    font=\sffamily},
  arn_n/.style = {treenode, circle, white, font=\sffamily\bfseries, draw=black,
    fill=black, text width=1.5em},
  arn_r/.style = {treenode, circle, black, draw=black, 
    text width=1.5em, very thick},
  arn_x/.style = {treenode, rectangle, draw=black,
    minimum width=0.5em, minimum height=0.5em}
}

\usepackage{tikz}
\usepackage{subfig}
\usetikzlibrary{positioning}
\usetikzlibrary{arrows}
\usepackage[english]{babel}
\usepackage{longtable}
\usepackage{blindtext}
\usepackage{float}
\usepackage{comment}
\usepackage{pifont}
\usepackage{color}
\definecolor{pastelred}{rgb}{1.0, 0.41, 0.38}
\definecolor{pistachio}{rgb}{0.58, 0.77, 0.45}
\newtheorem{observation}{Observation}
\newtheorem{open}{Open Problem}
\newtheorem{conclusion}{Conclusion} 

\usepackage[hyphens]{url}
\authorrunning{G. Istrate, C.  Bonchi\c{s}, and V.  Rochian}
\titlerunning{The language (and series) of Hammersley-type processes}

\begin{document}

\title{\textbf{The language (and series) of Hammersley-type processes}\thanks{This work was supported by a grant of
Ministry of Research and Innovation, CNCS - UEFISCDI, project number
PN-III-P4-ID-PCE-2016-0842, within PNCDI III.}
}

\author{ Cosmin Bonchi\c{s} \inst{1,2} \and Gabriel Istrate\inst{1,2} \and Vlad Rochian \inst{1}}

\institute{Department of Computer Science, West University of Timi\c{s}oara,  Bd. V. P\^{a}rvan 4,\\  Timi\c{s}oara,  Romania. Corresponding author's email: gabrielistrate@acm.org \and e-Austria Research Institute, Bd. V. P\^{a}rvan 4, cam. 045 B, Timi\c{s}oara,  Romania.}

\maketitle
\toctitle{THE LANGUAGE (AND SERIES) OF HAMMERSEY-TYPE PROCESSES}
\authorrunning{G.~Istrate, C.~Bonchi\c{s} and V. Rochian}
\tocauthor{Gabriel~Istrate \and Cosmin~Bonchi\c{s} \and Vlad~Rochian}

\begin{abstract}{\small 
We study languages and formal power series associated to (variants of) the Hammersley process. We show that the ordinary Hammersley process yields a regular language and the Hammersley tree process yields deterministic context-free (but non-regular) languages. For the Hammersley interval process we show that there are {\em two} relevant variants of formal languages. One of them leads to the same language as the ordinary Hammersley tree process. The other one yields non-context-free languages. 

The results are motivated by the problem of studying the analog of the famous Ulam-Hammersley problem for  heapable sequences. Towards this goal we also give an algorithm for computing formal power series associated to the Hammersley process. We employ this algorithm to settle the nature of the scaling constant, conjectured in previous work to be the golden ratio. Our results provide experimental support to this conjecture.}
\end{abstract}

\section{Introduction}

The Physics of Complex Systems and Theoretical Computing have a long and fruitful history of cooperation: for instance the celebrated Ising Model can be studied combinatorially, as some of its versions naturally relate to graph-theoretic concepts \cite{welsh:knots:cup}.  Methods from formal language theory have been employed (even in papers published by physicists, in physics venues) to the analysis of dynamical systems \cite{moore-lakdawala,xie-formal-languages}. 
Sometimes the cross-fertilization goes in the opposite direction: concepts from the theory of \textit{interacting particle systems} \cite{liggett-ips} (e.g. the voter model) have been useful in the analysis of gossiping protocols. A relative of the famous TASEP process, the so-called \textit{Hammersley-Aldous-Diaconis} (HAD) process, has provided \cite{aldous1995hammersley} the most illuminating solution to the famous {\it Ulam-Hammersley problem} \cite{romik2015surprising} concerning the scaling behavior of the longest increasing subsequence of a random permutation. 

In this paper we contribute to the literature on investigating physical models with discrete techniques by bringing methods based on formal language theory (and, possibly, noncommutative formal power series) to the analysis of several variants of the HAD process: We define formal languages (and power series) encoding all possible trajectories of such processes, and completely determine (the complexity of) these languages. 

The main process we are concerned with was defined in combinatorially in \cite{istrate2015heapable}, and in more general form in \cite{basdevant2016hammersley}, where it was dubbed \textit{the Hammersley tree process}. It appeared naturally in \cite{istrate2015heapable} as a tool to investigate a version of the Ulam-Hammersley problem that employs the concept (due to Byers et al. \cite{byers2011heapable})  of {\it heapable sequence}, an interesting variation on the concept of increasing sequence. Informally, a sequence of integers is heapable if it can be successively inserted into the leaves of a (not necessarily complete) binary tree satisfying the heap property. The Ulam-Hammesley problem for heapable sequences is open, the scaling behavior being the subject of an intriguing conjecture (see  Conjecture~\ref{conject} below) involving the golden ratio \cite{istrate2015heapable}. Methods based on formal power series can conceivably rigorously establish the true value of this  constant. We also study a (second) version of the Hammersley tree process, motivated by the analogue of the Ulam-Hammersley problem  for random intervals \cite{balogh2017computing} (see  Conjecture~\ref{conj:int} below).  

The outline of the paper is the following: In section~\ref{sec:first} we precisely specify the the systems we are interested in, and outline the results we obtain. In section~\ref{sec:motiv} we discuss the combinatorial and probability-theoretic motivations of the problems we are  interested in. This section is not needed to understand the technical details of our proofs. In Section~\ref{sec:main} we prove our main result: we precisely identify the Hammersley language for every $k\geq 1$. The language turns out to be regular for $k=1$ and deterministic context-free but non-regular for $k\geq 2$. The result is  then extended to (the analog of) the Hammersley process {\it for intervals}. In this case, it turns that there are \textit{two} natural ways to define the associated formal language. The "effective" version yields the same language as in the case of permutations. The "more useful" one yields (as we show) non-context-free languages that can be explicitly characterized. 
We then proceed by presenting (Section~\ref{sec:pseries}) algorithms for computing the power series associated to these systems. They are applied to the problem of determining true value of scaling constant (believed to be equal to the golden ratio) in the Ulam-Hammersley problem for heapable sequences. In a nutshell, \textbf{the experimental results tend to confirm the identity of this constant to the golden ratio}; however the convergence is slow, as the estimates based on the formal power series computations we undertake (based on small values of $n$) seem quite far from the true value. The paper concludes (Section~\ref{sec:open}) with several discussions and open problems. 


\section{Main definitions and results} 
\label{sec:first} 
We are interested in the following variant of the process in \cite{basdevant2016hammersley}, totally adequate for the purpose of describing the heapability of random permutations, defined as follows: 

\begin{definition} In the process $HAD_{k}$, individuals appear at integer times $t\geq 1$. Each individual can be identified with a \textrm{value} $X_{t}\in {\bf R}$, and is initially endowed with $k$ "lives". The appearance of a new individual $X_{t+1}$ subtracts a life from the smallest individual $X_{a}> X_{t+1}$ (if any) still alive at moment $t$. 
\end{definition} 

We can describe combinatorially the evolution of process $HAD_{k}$ in the following manner: each state of the system at a certain moment $n$ can be encoded by a word of length $n$ over the alphabet $\Sigma_{k}=\{0,1,\ldots, k\}$ obtained by discarding the value information from particles and only record the number of lives. Thus particles are arranged in the increasing order of values, from the smallest to the largest

\begin{example} 
Consider the state $s_{X}$ of the system $HAD_{2}$ after all particles with values $X=[5,1,4,2,3]$ have arrived (in this order). Then in the state $s_{X}$ particle 5 has 0 lives, particle 1 has two lives, particle 4 has 0 lives left, particle 2 has two lives, particle 3 has two lives left. Consequently, the word $w_{X}$ encoding $s_{X}$ is $22200$. 
\label{ex11}
\end{example}

Given this encoding, the dynamics of process $HAD_{k}$ on random permutations can be described in a completely equivalent manner as a process on words: given word $w_{k}$ encoding the state of the system at moment $k$, we choose a random position of $w_{k}$, inserting a $k$ there and subtract one from the first nonzero digit to the right of the insertion place, thus obtaining the word $w_{k+1}$.  This first nonzero digit need not be directly adjacent to the insertion place, but separated from it by a block of zeros. These zeros will not be affected in the word $w_{k+1}$. 
Figure~\ref{fig1} presents the snapshots of all possible trajectories of system $HAD_{2}$ at moments $t=1,2,3.$
 
 \begin{example}
If we run the process $HAD_{2}$ on sequence $X$ from Example~\ref{ex11}, the outcome is a multiset of particles $1,2$ and $3$, each with multiplicity 2, encoded by the word $22200$. 
\end{example}

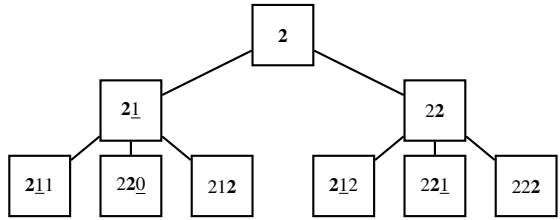
\begin{figure} 
\begin{center} 
\begin{tikzpicture}[
    scale=0.8, transform shape, thick,
    every node/.style = {draw, circle, minimum size = 10mm},
    grow = down,  
    level 1/.style = {sibling distance=5cm},
    level 2/.style = {sibling distance=1.5cm}, 
    level 3/.style = {sibling distance=1cm}, 
    level distance = 1.25cm
  ]

\node[shape=rectangle] {{\bf 2}}
    child{ node [shape=rectangle] {{\bf 2}{\underline 1}} 
    child{ node [shape=rectangle] {{\bf 2}{\underline 1}1}}
         child{ node [shape=rectangle] {2{\bf 2}{\underline 0}}}
         child{ node [shape=rectangle] {21{\bf 2}}}	
            }                    
    child{ node [shape=rectangle] {2{\bf 2}}
    	child{ node [shape=rectangle] {{\bf 2}{\underline 1}2}}
         child{ node [shape=rectangle] {2{\bf 2}{\underline 1}}}
         child{ node [shape=rectangle] {22{\bf 2}}}
         }
		
; 
\end{tikzpicture}
\end{center} 
\caption{Words in the Hammersley tree process ($k=2$). Insertions are boldfaced. Positions that lost a life at the current stage are underlined.}
\label{fig1} 
\end{figure}

We are interested in the following formal power series that encodes the large-scale evolution of process $HAD_{k}$, and the associated formal language: 

\begin{definition} Given $k\geq 1$, {\rm the Hammersley power series of order $k$} is the formal power series 
  $F_{k}\in {\bf N}(<\Sigma_{k}>)$ defined as follows: given word $w\in \Sigma_{k}^{*}$, define $F_{k}(w)$ to be the multiplicity of word $w$ in the process $HAD_{k}$. 

  The {\rm Hammersley language of order $k$,} $L_{H}^{k}$, is defined as the support of $F_{k}$, i.e. the set of words in $\Sigma_{k}^{*}$ s.t. there exists a {\rm trajectory} of $HAD_{k}$ that yields $w$. 
\end{definition} 

\begin{example}
$F_{2}(212)=2, F_{2}(220)=1$, hence $212,220\in L_{H}^{2}$. On the other hand $200\not\in L_{H}^{2},$ since $F_{2}(200)=0.$
\end{example} 

\begin{definition} For $w\in \Sigma_{k}^*$ and $a\in \Sigma_{k}$, denote by $|w|_a$ the number of copies of a in $w$.  Given $k\geq 1$, word $w\in \Sigma_{k}$ is called {\em $k$-dominant} if the following inequality holds {\bf for every $z\in Pref(w)$}:  
$|z|_{k} - \sum _{i=0}^{k-2} (k-i-1)\cdot |z|_{i}>0.$ 
We call the left-hand side term \textbf{ the structural difference} of word $z$. 
\end{definition} 

\begin{observation}
1-dominant words are precisely those that start with a 1. On the other hand, 2-dominant words are those that start with a 2 and have, in any prefix,  {\em strictly} more twos than zeros.

\end{observation} 

Our main result completely characterizes the Hammersley language of order $k$: 

\begin{theorem} \label{main}
For every $k\geq 1$, $
L_{H}^{k}=\{ w\in \Sigma_{k}^{*}\mbox{ }|\mbox{ } w\mbox{ is }k\mbox{-dominant.}\}.$
\end{theorem}

\begin{corollary}  Language $L_{H}^{1}$ is regular. For $k\geq 2$ languages $L_{H}^{k}$ are deterministic one-counter languages but not regular. 
\label{complexity} 
\end{corollary}

In \cite{balogh2017computing} we considered the extension of heapability to partial orders, including intervals. We also noted that, just as in the case of random permutations, heapability of random intervals can be analyzed using the following version of the process $HAD_{k}$:

\begin{definition} The {\em interval Hammersley process with $k$ lives} is the stochastic process defined as follows: The process starts with no particles.  Particles arrive at integer moments; they have a {\it value} in the interval $(0,1)$, and a number of {\it lives}. Given the state $Z_{n-1}$ of the process after step $n-1$, to obtain $Z_{n}$ we choose, independently, uniformly at random and with repetitions two random reals $X_{n},Y_{n}\in (0,1)$. Then we perform the following operations: 
\begin{itemize} 

\item[-] First a new particle with $k$ lives and value $min(X_{n},Y_{n})$ is inserted. 
\item[-] Then the smallest (if any) live particle whose value is higher than  $max(X_{n},Y_{n})$ loses one life,   yielding state $Z_{n}$. 
\end{itemize} 
The {\em state} of the process at a certain moment $n$ comprises a record of all the real numbers chosen along the trajectory:, $(X_{0},Y_{0}, \ldots, X_{n-1},Y_{n-1})$, even those that do not correspond to a particle. Each number is endowed (in case it represented a new particle) with an integer in the range $0\ldots k$ representing the number of lives the given particle has left at moment $n$. 
\label{ham:int} 
\end{definition}

\noindent Just as with process $HAD_{k}$, we can combinatorialize the previous definition as follows: 

\begin{definition} Process {\em$HAD_{k,INT}$} is the stochastic process on $(\Sigma_{k}\cup \{\diamond\})^{*}$ defined as follows:  The process starts with the two-letter word $Z_{1}=k\diamond$. Given the string representation $Z_{n-1}$ of the process after step $n-1$, we choose, independently, uniformly at random and with repetitions two positions $X_{n},Y_{n}$ into string $Z_{n-1}$. $X_{n},Y_{n}$ may happen to be the same position, in which we also choose randomly an ordering of $X_{n},Y_{n}$.  Then we perform (see Figure~\ref{fig:17}) the following operations: 
\begin{itemize} 
\item First, a $k$ is inserted into $Z_{n-1}$ at position $min(X_{n},Y_{n})$. 
\item Then a $\diamond$ is introduced in position $max(X_{n},Y_{n})$ (immediately after the newly introduced $k$, if $X_{n}=Y_{n}$). 
\item Then the smallest (if any) nonzero digit occurring \textbf{after the position of the newly inserted $\diamond$} loses one unit. This yields string $Z_{n}$. 
\end{itemize} 
\label{hamint}
\end{definition}

\begin{figure}[!ht]
\begin{center} 
		\begin{tikzpicture}[scale=0.8,->,>=stealth',shorten >=1pt,auto,node distance=0.35cm,	thick,main node/.style={rectangle,font=\sffamily\small\bfseries}]
		
		\node[main node] (x) {$X_{n}$};
		\node[main node] (r1-2) [below of = x, node distance = 1cm] {$\cdot$};
		\node[main node] (r1-1) [left = of r1-2] {$2$};
		\node[main node] (r1-0) [left = of r1-1] {$\cdot$};	
		\node[main node] (r1-0-0) [left = of r1-0] {$Z_{n-1}=$};	
		\node[main node] (r1-3)  [right = of r1-2] {$2$};
		\node[main node] (r1-4) [right = of r1-3] {$\cdot$};	
		\node[main node] (r1-5) [right = of r1-4] {$0$};
		\node[main node] (r1-6) [right = of r1-5] {$\cdot$};		
		\node[main node] (r1-7) [right = of r1-6] {$2$};
		\node[main node] (r1-8) [right = of r1-7] {$\cdot$};
		\node[main node] (r1-9) [right = of r1-8] {$0$};
		\node[main node] (r1-10) [right = of r1-9] {$\cdot$};		
		
		\node[main node] (x1) [right = of x] {$$};
		\node[main node] (x2) [right = of x1] {$$};
\node[main node] (x3) [right = of x2] {$~~~$};		
		\node[main node] (y) [right = of x3] {$Y_{n}$};

		\node[main node] (min) [below of = r1-2, node distance = 1cm] {$min\{X_{n},Y_{n}\}$};
		\node[main node] (max) [below of = r1-6, node distance = 1cm] {$max\{X_{n},Y_{n}\}$};
		\node[main node] (r1-4mid) [below of = r1-4, node distance = 1cm] {};
		
		\node[main node] (r2-2) [below of = min, node distance = 1cm] {$\mathbf{2}$};
		\node[main node] (r2-1) [left = of r2-2] {$2$};
		\node[main node] (r2-0) [left = of r2-1] {$\cdot$};	
		\node[main node] (r2-0-0) [left = of r2-0] {$Z_{n}=$};			
		\node[main node] (r2-3)  [right = of r2-2] {$2$};
		\node[main node] (r2-4) [right = of r2-3] {$\cdot$};	
		\node[main node] (r2-5) [right = of r2-4] {$0$};
		\node[main node] (r2-6) [right = of r2-5] {$\diamond$};		
		\node[main node] (r2-7) [right = of r2-6] {$\mathbf{1}$};
		\node[main node] (r2-8) [right = of r2-7] {$\cdot$};
		\node[main node] (r2-9) [right = of r2-8] {$0$};
		\node[main node] (r2-10) [right = of r2-9] {$\cdot$};			
		
		\draw[->] (x) -- (r1-2);
		\draw[->] (y) -- (r1-6);
		\draw[->] (min) -- (r2-2);
		\draw[->] (max) -- (r2-6);
		
		\draw[double, ->] (r1-4) --(r1-4mid);
		
		
		\end{tikzpicture}
\end{center} 
	\caption{Insertion in process $HAD_{2,INT}$. Insertion positions are marked with a dot. Positions affected by the insertion are in bold.}
	\label{fig:17}
	
\end{figure}

It turns out (see the discussion at the end of  Section~\ref{sec:motiv}) that there are \textit{two} languages meaningfully associated to the process $HAD_{k,INT}$. The first of them has the following definition: 

\begin{definition} 
Denote by $L_{H, INT}^{k}$, called {\rm the language of the interval Hammersley process}, the set of words (over alphabet $\Sigma_{k}\cup \{\diamond \}$) generated by the  process $HAD_{k,INT}$. 
\label{lang:ham:int}
\end{definition} 

\noindent The second language associated to the interval Hammersley process is defined as follows: 

\begin{definition} 
 The {\rm effective language of the interval Hammersley process}, $L_{H, INT}^{k,\mbox{eff}}$, is the set of strings in $\Sigma_{k}^{*}$ obtained by deleting all diamonds from some string in $L_{H,INT}^{k}$. 
 \label{lang:2}
\end{definition}

Despite the fact that the dynamics of  process $HAD_{k,INT}$ is quite different from that of the ordinary process $HAD_{k}$ (a fact that is reflected in the coefficients of the two power series), and the conjectured scaling behavior is not at all similar (for $k\geq 2$), our next result shows that this difference {\bf is not visible} on the actual trajectories: the effective language of the Interval Hammersley process coincides with that of the "ordinary" Hammersley tree process. Indeed, we have:  

\begin{theorem} 
For every $k\geq 1$, $L_{H, INT}^{k,\mbox{eff}}=L_{H}^{k}=\{ w\in \Sigma_{k}^{*}\mbox{ }|\mbox{ } w\mbox{ is }k\mbox{-dominant.}\}.$
\label{foo2}
\end{theorem}

The previous result contrasts with our next theorem: 

\begin{theorem}
For $k\geq 1$ the language $L_{H,INT}^{k}$ is {\bf not} context-free. 
\label{noncfl} 
\end{theorem} 

In fact we can give a complete characterization of $L_{H,INT}^{k}$ similar in spirit to the one given for language $L_{H}^{k}$ in Theorem~\ref{main}: 

\begin{theorem} Given $k\geq 1$, the language $L_{H,INT}^{k}$ is the set of words $w$ over alphabet $\Sigma_{k}\cup \{ \diamond\}$ that satisfy the following conditions: 
\begin{enumerate} 
\item $|w|_{\diamond}=|w|/2$. In particular $|w|$ must be even. 
\item For every prefix $p$ of $w$, (a).   $|p|_{\diamond}\leq |p|/2$ and (b). $s(p) + (k+1)|p|_{\diamond}\geq k|p|.$ 
\end{enumerate} 
\label{foo3} 
\end{theorem} 

Finally, we return to the power series perspective on the Ulam-Hammersley problem for heapable sequences. We outline a simple algorithm (based on dynamic programming) for computing the coefficients of the Hammersley power series $F_{k}$. 

{\small
\begin{figure}[h]
\begin{center}
	\fbox{
	    \parbox{11cm}{

\textbf{Input:} $k\geq 1,w\in \Sigma_{k}^{*}$ \\
\textbf{Output:} $F_{k}(w)$
\vspace{-2mm}
\begin{itemize} 
\item[] $S:= 0$. $w=w_{1}w_{2}\ldots w_{n}$
\item[] \textbf{if }$w\not \in L_{H}^{k}$
\begin{itemize} 
   \item[]\textbf{return }0 
\end{itemize} 
\textbf{if } $w== `k`$
   \begin{itemize} 
   \item[] \textbf{return }1 
   \end{itemize} 
\textbf{ for } $i\mbox{ in 1:n-1}$
   \begin{itemize} 
       \item[] \textbf{if } $w_{i}==k\mbox{ and } w_{i+1} \neq k$
        \begin{itemize}
               \item[] \textbf{let }$r=min\{l\geq 1: w_{i+l}\neq 0\mbox{ or }i+l=n+1\}$
               \item[] \textbf{for }$j\mbox{ in 1:r-1}$
               \begin{itemize} 
                       \item[] \textbf{let }$z=w_{1}\ldots w_{i-1}w_{i+1}\ldots w_{i+j-1}1w_{i+j+1}\ldots         w_{i+r}\ldots w_{n}$ 
                       \item[] $S := S +  ComputeMultiplicity(k,z)$ 
              \end{itemize} 
              \item[] \textbf{if }$i+r\neq n+1$ and $w_{i+r}\neq k$
              \begin{itemize} 
                      \item[]\textbf{let }$z=w_{1}\ldots w_{i-1}w_{i+1}\ldots w_{i+r-1} (w_{i+r}+1)w_{i+r+1}\ldots w_{n}$  
                       \item[] $S := S +  ComputeMultiplicity(k,z)$
             \end{itemize} 
       \end{itemize} 
\end{itemize} 
\begin{itemize} 
  \item[]\textbf{if }$w_{n}== k$
  \begin{itemize}
      \item[] \textbf{let }$Z=w_{1}\ldots \ldots w_{n-1}$  
      \item[] $S := S +  ComputeMultiplicity(k,z)$ 
  \end{itemize} 
\end{itemize} 
\item[] \textbf{return }S
\end{itemize} 
}
}
\end{center}
\vspace{-0.5cm}
\caption{Algorithm ComputeMultiplicity(k,w)} 
\label{algo} 
\end{figure}}

\begin{theorem}
Algorithm ComputeMultiplicity correctly computes series $F_{k}$. 
\label{foo4} 
\end{theorem} 

We defer the presentation of the application of this result to Section~\ref{sec:apps}. 
\section{Motivation and notations}
\label{sec:motiv}

Define $\Sigma_{\infty}=\cup_{k\geq 1} \Sigma_{k}$. Given $x,y$ over $\Sigma_{\infty}$ we use notation $x\sqsubseteq y$ to denote the fact that $x$ is a prefix of $y$. The set of (non-empty) prefixes of $x$ is be denoted by $Pref(x)$.

A {\it $k$-ary (max)-heap} is a $k$-ary tree, non necessary complete, whose nodes have labels $t[\cdot]$ respecting the min-heap condition $t[{parent(x)]} \geq t[x].$
Let $rgeq 1$, and let $a_{1},b_{1},\ldots, a_{r}>0$ and $b_{r}\geq 0$ be integers. We will use notation $[a_{0},b_{0},\ldots, a_{t},b_{t}]$ as a shorthand for the word $1^{a_{0}}0^{b_{0}}\ldots 1^{a_{r}}0^{b_r}\in \Sigma_{1}^{*}$ (where $0^0=\epsilon$, the null word). 

The following combinatorial concept was introduced (for $k=2$) in \cite{byers2011heapable} and further studied in \cite{istrate2015heapable,heapability-thesis,istrate2016heapability,basdevant2016hammersley,basdevant2017almost,balogh2017computing}: 

\begin{definition} 
  A  sequence $X=X_{0},\ldots, X_{n-1}$ is {\it max $k$-heapable} if there exists
  some $k$-ary tree $T$ with nodes labeled by (exactly one of) the elements of $X$,
  such that for every non-root node $X_{i}$ and parent $X_{j}$, $X_{j}\geq X_{i}$
  and $j<i$. In particular a 2-heapable sequence will simply be called
  {\it heapable}  \cite{byers2011heapable}. Min heapability is defined similarly. 
\end{definition}  

\begin{example}
  $X=[5,1,4,2,3]$ is max 2-heapable: A max 2-heap is displayed in
  Figure~\ref{fig11}. On the other $Y=[2,4,1,3]$ is obviously
  {\bf not} max 2-heapable, as 4 cannot be a descendant of  2.\
  \label{ex1}
\end{example} 

\begin{figure} 
\begin{center} 
\begin{tikzpicture}[scale=0.8, ->,>=stealth',level/.style={sibling distance = 3cm/#1,
  level distance = 1cm}] 
\node [arn_r] {5}
    child{ node [arn_r] {1} 
 		child{ node [arn_x] {}}
        child{ node [arn_x] {}}
            }                         
    child{ node [arn_r] {4}
    	child{ node [arn_r] {2}
            child{ node [arn_x] {}} 
             child{ node [arn_x] {}} 
            }
            child{ node [arn_r] {3}
							child{ node [arn_x] {}}
							child{ node [arn_x] {}}
            }
		}
; 
\end{tikzpicture}
\end{center} 

\caption{Heap ordered tree for sequence X in Example~\ref{ex1}.}
\label{fig11} 
\end{figure}
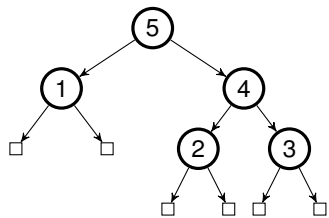 

Heapability can be viewed as a relaxation of the notion of decreasing sequence,
thus it is natural to attempt to extend to heapable sequences the framework of the
Ulam-Hammersley problem \cite{romik2015surprising}, concerning the 
scaling behavior of the longest increasing subsequence (LIS) of a random
permutation. This extension can be performed in (at least) two ways, equivalent
for LIS but no longer equivalent for heapable sequences: the first way, that of
studying the length of the longest heapable subsequence, was dealt with in
\cite{byers2011heapable}, and is reasonably simple: with high probability the
length of the longest heapable subsequence of a random permutation is $n-o(n)$. On the other hand, by Dilworth's theorem \cite{dilworth1950decomposition} the length of the longest increasing subsequence of an arbitrary sequence is equal to the number of classes in a partition of the original sequence into \textit{decreasing} subsequences. Thus it is natural to call {\em the Hammersley-Ulam problem for heapable sequences} the investigation of the scaling behavior of the number of classes of the partition of a random permutation into a minimal number of (max) heapable subsequences. This was the approach we took in \cite{istrate2015heapable}. Unlike the case of LIS, for heapable subsequences the relevant parameter (denoted in \cite{istrate2015heapable} by $MHS_{k}(\pi)$) scales logarithmically, and the following conjecture was proposed: 

\begin{conjecture} 
For every $k\geq 2$ there exists $\lambda_{k}>0$ s.t., as $n\rightarrow \infty,$ $\frac{E[MHS_{k}(\pi)]}{ln(n)}$ converges to $\lambda_{k}$. Moreover {\bf $\lambda_{2}=\frac{1+\sqrt{5}}{2}$ is the golden ratio. }
\label{conject} 
\end{conjecture} 

The problem was further investigated in \cite{basdevant2016hammersley,basdevant2017almost}, where the existence of the constant $\lambda_{k}$ was proved. The equality of $\lambda_{2}$ to the golden ratio is less clear: authors of \cite{basdevant2016hammersley} claim it is slightly less than $\phi$. Some non-rigorous, "physics-like" arguments, in favor of the identity $\lambda_{2}=\phi$ was already outlined in \cite{istrate2015heapable}, and is presented in \cite{istrate2016heapability}, together with experimental evidence. Here we bring more convincing such evidence. 

The intuition for Conjecture~\ref{conject} relies on the extension from the LIS problem to heapable sequences of a correspondence between LIS and an interactive particle system \cite{aldous1995hammersley} called the  {\it Hammersley-Aldous-Diaconis (shortly, Hammersley or HAD) process}. The validity of correspondence was noted, for heapable sequences, in \cite{istrate2015heapable}. The generalized process was further investigated in \cite{basdevant2016hammersley}, where it was called {\em the Hammersley tree process}.

 To recover the connection with random permutations we will assume from now on that the $X_{i}$'s in process $HAD_{k}$ are independent random numbers in $(0,1)$. The proposed value for $\lambda_{2}$ arises from a conjectural identification of the "hydrodynamic limit" of the Hammersley tree process (in the form of a compound Poisson process).

 As $n\rightarrow \infty$ a "typical" sample word from the Hammersley process $HAD_{2}$ will have approximately $c_{0}n$ zeros, $c_{1}n$ ones and $\sim c_{2}n$ twos, for some constants\footnote{Nonrigorous computations predict that $c_{0}=c_{2}=\frac{\sqrt{5}-1}{2}, c_{1}=\frac{3+\sqrt{5}}{2}.$} $c_{0},c_{1},c_{2}>0$.  Moreover, conditional on the number of zeros, ones, twos, in a typical word these digits are "uniformly mixed" throughout the sequence.  Experimental evidence presented in \cite{istrate2016heapability} seems to confirm the accuracy of this heuristic description. 

A proof of the existence of constants $c_{0},c_{1},c_{2}$ was attempted in \cite{istrate2015heapable} based on subadditivity (Fekete's lemma). However, part of the proof in \cite{istrate2015heapable} is incorrect. While it could perhaps be fixed using more sophisticated  tools (e.g. the subadditive ergodic theorem \cite{spa-average}) than those in \cite{istrate2015heapable}, an alternate approach involves analyzing the asymptotic behavior of process $HAD_{k}$ using (noncommutative) power series (\cite{salomaa2012automata,berstel2011noncommutative}). 

Understanding and controlling the behavior of formal power series $F_{k}$ may be the key to obtaining a rigorous analysis that confirms the picture sketched above. Though that we would very much want to accomplish this task, in this paper we resign ourselves to a simpler, language-theoretic, version of this problem, that of computing the associated formal language.

The Ulam-Hammersley problem has also been studied \cite{justicz1990random} for sets of random intervals, generated as follows: to generate a new interval $I_{n}$ first we sample (independently and uniformly) two random $x,y$ from $(0,1)$. Then we let $I_{n}$ be the interval $[min(x,y), max(x,y)]$.  In fact the problem was settled in \cite{justicz1990random}, where the scaling of LIS for sets of random intervals was determined to be $ 
\lim_{n\rightarrow \infty} \frac{E[LIS(I_{1},\ldots, I_{n})]}{\sqrt{n}}=\frac{2}{\sqrt{\pi}}.$

 Several results on the heapability of partial orders were proved in \cite{balogh2017computing}; in particular, the greedy algorithm for partitioning a permutation into a minimal number of heapable subsequences extends to interval orders. This justifies an extension of the Ulam-Hammersley problem from increasing to heapable sequences of intervals.  Indeed, in  \cite{balogh2017computing} we conjectured the following scaling law: 

\begin{conjecture} 
 For every $k\geq 2$ there exists $c_{k}>0$ such that, if $R_{n}$ is a sequence of $n$ random intervals then $
\lim_{n\rightarrow \infty} \frac{E[\#Heaps_{k}(R_n)]}{n}=c_{k}.\mbox{ Moreover }c_{k}=\frac{1}{k+1}.$
\label{conj:int} 
\end{conjecture} 

Remarkably, it was already noted in \cite{balogh2017computing} that the connection between the Ulam-Hammersley problem and particle systems extends to the interval setting as well. To prove a similar result for the interval Hammersley process we need to "combinatorialize" the process from Definition~\ref{ham:int}, that is, to replace that definition (which employs (random) real values in $(0,1)$) with an equivalent stochastic process on words. 

The combinatorialization process has some technical complications with respect to the case of permutations. Specifically, for permutations the state of the system could be preserved, with no real loss of information by a string representing only the number of lifelines of the given particles, but \textbf{not} their actual values. 
This enables (as we will see below in Section~\ref{sec:pseries}) an algorithm for computing the associated formal power series. 

To accomplish a similar goal for random intervals we apparently need to take into account the fact that at each step we choose {\em two} random numbers in Definition~\ref{ham:int}, even though only one of them receives a particle, since the second one influences the state of the system. Thus, the proper discretization requires an extra symbol $\diamond$ (that marks the positions of real values that were generated but in which no particle was inserted), and is accomplished as described in Definition~\ref{hamint} 
and the language from Definition~\ref{lang:ham:int}.  

A result that was easy for the process $HAD_{k}$ but deserves some discussion in the case of the interval process is the following: 

\begin{proposition} 
Consider the string $w_{n}\in \Sigma_{k}^{*}$ obtained by taking a random state of the Hammersley interval process with $k$ lifelines at stage $n$ and then "forgetting" the particle value information (recording instead only the value in $\Sigma_{k}\cup \{\diamond\}$). Then $w_{n}$ has the same distribution as a sample from process $HAD_{k,INT}$ at stage $n$.  
\end{proposition} 
\begin{proof} 

The crux of the proof is the following 

\begin{lemma} 
The ordering of the values $X_{0},Y_{0},X_{1},Y_{1},\ldots, X_{n-1}Y_{n-1}$ inserted in the first $n$ steps in the Hammersley interval process (disregarding their number of lifelines) is that of a random permutation with $2n$ elements. 
\end{lemma} 
\begin{proof}
$X_{i},Y_{i}$ have the same distribution, both are random uniformly distributed variables in $(0,1)$.  Thus to simulate  $HAD_{k,INT}$ for $n$ steps one needs $2n$ random numbers in (0,1), which yields a random permutation of size $2n$. 
\end{proof}
\end{proof} 

This discussion motivates the language-theoretic study of trajectories of the interval Hammersley process $HAD_{k,INT}$ as well. In that respect Definition~\ref{lang:2} seems better motivated than Definition~\ref{lang:ham:int}. Indeed, due to the presence of diamonds, words in the Definition~\ref{lang:2} are not "physical", as diamonds do not necessarily correspond to actual particles. On the other hand one can easily obtain an algorithm (similar to the ComputeMultiplicity algorithm presented above) that computes multiplicities for 
``extended words" in the process $HAD_{k,INT}$ such as those in the Definition~\ref{lang:ham:int}. Hence the study of this second language is motivated on pragmatic grounds, as a first step to investigating $F_{k, INT}$, the formal power series of multiplicities in the interval Hammersley process. We defer this investigation to the journal version of the paper.

\section{Proof  of the main result} 
\label{sec:main}

The proof of Theorem~\ref{main} proceeds by double inclusion. Inclusion "$\subseteq$" is proved with the help of several easy auxiliary results: 

\begin{lemma} \label{claim:startWith2}
Every word in $L_H^k$ starts with a $k.$
\end{lemma}

\begin{proof}
Follows easily by appealing to the particle view of the Hammersley process: the particle with the smallest label $x$ stays with $k$ lives until the end of the process, as no other particle can arrive to its left. 
\end{proof}

\begin{lemma}\label{claim:prefixIsWord}
$L_H^k$ is closed under prefix.
\end{lemma}

\begin{proof}
Again we resort to the particle view of the Hammersley process: let $w\in L_{H}^{k}$ be a word and $u=x_{0}\ldots x_{n-1}$ be a trajectory in [0,1] yielding $w$. A non-empty prefix $z$ of $w$ corresponds to the restriction of $u$ to some segment $[0,l]$, $0<l<1$. This restriction is a trajectory itself, that yields $z$. 
\end{proof}

\begin{lemma}\label{claim:hasMore2}
Every word in $L_H^k$ has a positive structural difference.  
\end{lemma}

\begin{proof}

Let $w\in L_{H}^{k}$ and let $t$ be a corresponding trajectory in the particle process. 

Let $\lambda$ be the number of times a particle arrives as a local maximum (without subtracting a lifeline from anyone). 
For $i=1,\ldots, k$ let $\lambda_{i}$ be the number of time the newly arrived particle subtracts a lifeline from a particle currently holding exactly $i$ lives. $\lambda, \lambda_{1},\ldots, \lambda_{k}\geq 0$. Moreover, $\lambda > 0$, since the largest particle does not take any lifeline. 

By counting the number of particles with $i$ lives at the end of the process, we infer: $|z|_{0}=\lambda_{1}, |z|_{1}=\lambda_{2}-\lambda_{1}, \cdots |z|_{k-1}=\lambda_{k}-\lambda_{k-1}.$ Finally, $
|z|_{k}= \lambda + \sum_{i=0}^{i-2} \lambda_{i}. (*)$ 

Simple computations yield $\lambda_{i+1}=|z|_{0}+\ldots + |z|_{i}$, for $i=0, \ldots, k-1$. Relation (*) and inequality $\lambda > 0$ yield the desired result. 
\end{proof}

Together, Claims~\ref{claim:startWith2},~\ref{claim:prefixIsWord} and~\ref{claim:hasMore2} establish the fact that any word from $L_H^k$ is $k$-dominant, thus proving inclusion "$\subseteq$". To proceed with the opposite inclusion,  for every $k$-dominant word $w$ we must construct a trajectory of the process $HAD_{k}$ that acts as a witness for $w\in L_{H}^k$.

We will further reduce the problem of constructing a trajectory $T_{z}$ to the case when $z$ further satisfies a certain simple property, explained below: 
\begin{definition}
$k$-dominant word  $u$ is called \textit{critical} if 
$|u|_k - \sum\limits_{i=0}^{k-2} (k-1-i)\cdot |u|_i = 1.$
\end{definition}


 The above-mentioned reduction has the following statement: 

\begin{lemma} \label{lemma:semiInL2} Every 
$k$-dominant {\bf critical} $z$ is witnessed by some trajectory $T_{z}$. 
\end{lemma}
\begin{proof} 
By induction on $|z|$. The base case, $|z| = 1$, is trivial, as in this case $z = k$.

{\textbf{Inductive step}}: Assume the claim is true for all the critical words of length strictly smaller than $z$'s. We claim that $w_1$, the word obtained from $z$ by deleting the last copy of $k$ and increasing by 1 the value of the letter immediately to the right of the deleted letter, is critical. 

Indeed, it is easy to see that the structural difference of $w_{1}$ is 1.  Clearly the deleted letter  could not have been the last one,  otherwise deleting it would yield a prefix of $z$ that has structural constant equal to zero. Also clearly, the letter whose value was modified in the previous constraint could not have been a $k$, by definition, and certainly is nonzero after modification. So $w_{1}$'s construction is indeed correct. 
As $|w_1|=|z|-1$, $w_{1}$ satisfies the conditions of the induction hypothesis. 

By the induction hypothesis, $w_1$ can be witnessed by some trajectory $T$. We can construct a trajectory for  $z$ by simply following $T$ and then inserting the last $k$ of $z$ into $w_{1}$ in  its proper position (thus also making the next letter assume the correct value).

\end{proof}

We now derive  Theorem~\ref{main} from Lemma~\ref{lemma:semiInL2}. The key observation is the following fact: every $k$-dominant word $z$ is a prefix of a {\em critical word}, e.g.  $z^{\prime}=z(k-2)^{\lambda}$ where 
$\lambda = |z|_{k}- \sum _{i=0}^{k-2} (k-i-1)\cdot |z|_{i} - 1\geq 0.$

By Lemma~\ref{lemma:semiInL2}, $z^{\prime}$ has a witnessing trajectory $T_{z^{\prime}}$. 
Since the existence of a trajectory is closed under taking prefixes, Theorem~\ref{main} follows. 

\section{Proof of Corollary~\ref{complexity}}

\begin{proof}

For $k=1$ the result is trivial, as $L_{H}^{1}=1\Sigma_{1}^{*}$. The claim that $L^{k}_{H}$ is a deterministic one-counter language for $k\geq 2$ follows from Theorem~\ref{main}, as one can construct a  one-counter pushdown automaton $P_{k}$ for  the language on $k$-dominant words. 

The one-counter PDA has input alphabet ${0, 1, 2, \ldots, k}$. Its stack alphabet contains two special stack symbols, the bottom symbol  $Z$ and another "counting" symbol $*$. The transitions of $P_{k}$ are informally defined as follows: 
\begin{itemize}
\item[-] $P_{k}$ starts with the stack consisting of the symbol $Z$. If the first letter is not a $k$, $P_{k}$ immediately rejects. Otherwise it pushes a $*$ on the stack. 
\item[-] on reading  any subsequent $k$, $P_{k}$ pushes a $*$ symbol on stack. 
\item[-] on reading any symbol $i\in 1\ldots k-2$, $P_{k}$ attempts to pop $k-i-1$ stars from the stack. If this ever becomes impossible (by reaching $Z$), $P_{k}$ immediately rejects. 
\item[-] $P_{k}$ ignores all $k-1$ symbols, proceeding without changing the content of the stack. 
\item[-] If, while reaching the end of the word, the stack still contains a star, $P_{k}$ accepts. 
\end{itemize}

To prove that $L_{H}^{k}$, $k\geq 2$, is not regular is a simple exercise in formal languages. 
 It involves applying the pumping lemma for regular languages to words  $w_{k,n}=k^{n(k-1)+1}0^{n}\in L_{H}^{k}$. We infer that for large enough $n$, $w_{k,n}=w_{1}w_{2}w_{3}$, with $w_{2}$ nonempty and consisting of $k$'s only, such that for every $l\geq 0$, $w_{1}w_{2}^{l}w_{3}\in L_{H}^{k}$. We obtain a contradiction by letting $l=0$, thus obtaining a word $z$ that cannot belong to $L_{H}^{k}$, since $|z|_{k}\leq (k-1)|z|_0$. 

\end{proof}\qed

\section{Proof of Theorem~\ref{foo2}}
It is immediate that $L_{H}^{k}\subseteq L_{H,INT}^{k,\mbox{eff}}$. Indeed, every trajectory of the process $HAD_{k}$ is a trajectory of the process $HAD_{k,INT}$ as well: simply restrict at every stage the two particles to choose the same slot. 

For the opposite inclusion we prove, by induction on $|t|$, that the outcome $w$ of every trajectory $t$ of the interval Hammersley process belongs to $L_{H}^{k}$. The case $|t|=0$ is trivial, since $w=k$.
\begin{definition} 
Given a word $w$ over $\Sigma_{k}$, word $z$ is a {\em left translate of $w$} if $z$ can be obtained from $z$ by moving a $k$ in $w$ towards the beginning of $w$ (we allow "empty moves", i.e. $z=w$). 
\end{definition} 
\begin{lemma} 
$L_{H}^{k}$ is closed under left translates. That is, if $w\in L_{H}^{k}$ and $z$ is a left translate of $w$ then $z\in L_{H}^{k}$. 
\end{lemma} 
\begin{proof} 
By moving forward a $k$ the structural constants of all prefixes of $w$ can only increase. Thus if these constants are positive for all prefixes of $w$ then they are positive for all prefixes of $z$ as well. 
\end{proof} 

Now assume that the induction hypothesis is true for all trajectories of length less than $n$. Let $t$ be a trajectory of length $n$, let $t^{\prime}$ be its prefix of length $n-1$, let $w$ be the yield of $t$ and $z$ be the yield of $t^{\prime}$. By the induction hypothesis $z\in L_{H}^{k}$. Let $y$ be the word obtained by applying the Hammersley process to $z$, deleting a life from the same particle as the interval Hammersley process does to $z$ to obtain $w$. It is immediate that $w$ is a left translate of $y$ (that is because in the interval Hammersley process we insert a particle to the left of the position where we would in $HAD_{k}$). Since $y\in L_{H}^{k}$, by the previous lemma $w\in L_{H}^{k}$.

\section{Proof of Theorem~\ref{noncfl}}

Define the language $S_{k}=L_{H,INT}^{k}\cap \{k\}^{*}\diamond^{*}\{k-1\}^{*}\diamond^{*}$. 
\begin{lemma} 
$S_{k}=\{ k^{c+d+e} \diamond^{c+e} (k-1)^{c}\diamond^{c+d}|c,d,e\geq 0\}$. 
\end{lemma} 
\begin{proof} 
The direct inclusion is fairly simple: let $w\in S_{k}$. define $c$ to be the number of letters $(k-1)$ in $w$. 
Since there are no diamonds in between the $(k-1)$'s, all such letters must have been produced by removing one lifeline each by some $k$'s. Thus the number of stars in between  the $k$'s and $(k-1)$'s is $c+e$, with $e$ being the number of pairs $(k,\diamond)$ that did not kill any particle that will eventually become a $k-1$. 

On the other hand the number of $k$'s is obtained by tallying up $c$ (for the $c$ letters that become $k-1$, needing one copy of $k$ each), $e$ (for the pairs $(k,\diamond)$ where $\diamond$ belongs to the first set of diamonds) and $d$ (for $d$ pairs $(k,\diamond)$ with $\diamond$ in the second set of diamonds). 

For the reverse implication we outline the following construction: 

\noindent
First we derive $k^{e}\diamond^{e}$. Then we repeat the following strategy $c$ times: 
\vspace{-3mm}
\begin{itemize} 
\item[-] We insert a $k$ at the beginning of the $k-1$ block (initially at the end of the word) and the corresponding $\diamond$ at the end of the word. 
\item[-] With one pair $k,\diamond$ (with $\diamond$ inserted in the first block) we turn the $k$ into a $k-1$. 
\end{itemize} 
\vspace{-2mm}
Finally we insert $k$ pairs $(k,\diamond)$, with $\diamond$ in the second block. 
\end{proof} 

The theorem now follows from the following 
\begin{lemma} 
$S_{k}$ is not a context-free language. 
\label{ogden}
\end{lemma} 
\begin{proof} 
An easy application of Ogden's lemma: 
We take a string $s\in S_k$, 
\[
s=k^{c+d+e} \diamond^{c+e} (k-2)^{c}\diamond^{c+d}
\]
with $c,d,e\geq p$ (where $p$ is the parameter in Ogden's Lemma. We mark all positions of $k-1$. 
Then $s=uvwxy$, with $uv^iwx^iy \in S_{k}$ for all $i\geq 0$. The "pumping blocks" $v,x$ cannot consist of more than one type of symbols, otherwise the pumped strings would fail to be a member of $\{k\}^{*}\diamond^{*}\{k-1\}^{*}\diamond^{*}$. 

Therefore no more than two blocks (of the four in $s$) get pumped. One that definitely gets pumped is the 
first block of diamonds. Taking large enough $i$ we obtain a contradiction, since the block that fails to get pumped will eventually have smaller length than the (pumped) first block of diamonds. 
\end{proof}

\section{Proof of Theorem~\ref{foo3}}
\begin{proof} 
The inclusion $\subseteq$ is easy: given $w\in L_{H,INT}^{k}$, conditions 1. and 2 (a). hold, as the process $HAD_{INT}^{k}$ inserts a digit (more precisely a $k$) before every diamond. 

As for condition 2 (b)., each $\diamond$ takes at most one life of a particle.  The total number of lives particles in $p$ are endowed with at their moments of birth is $k(|p|-|p|_{\diamond})$. These lives are either preserved (and are counted by $s(p)$), or they are lost, in a move which (also) introduces a $\diamond$ in $p$. Thus $
k(|p|-|p|_{\diamond}) \leq |p|_{\diamond} + s(p),$
which is equivalent to b. 

The inclusion $\supseteq$ is proved by induction on $|w|$. What we have to prove is that every word that satisfies conditions 1-2 is an output of the process $HAD_{k,INT}$. 

The case $|w|=2$ is easy: the only word that satisfies conditions 1-2 is easily seen to be $w=k\diamond$, which can be generated in one move. 

Assume now that the induction hypothesis is true for all words of lengths strictly less than 2n, and let $w=w_{1}\ldots w_{2n}$ be a word of length $2n$ satisfying conditions 1-2. 

\begin{lemma} 
$w_{2n}=\diamond.$
\end{lemma} 
\begin{proof} 
Let $p=w_{1}\ldots w_{2n-1}$. By condition 2(a)., $|p|_{\diamond}\leq (2n-1)/2$, hence $|p|_{\diamond}\leq n-1$. Since $|w|_{\diamond}=n$, the claim follows. 
\end{proof} 

\begin{lemma} 
$w_{1}=k$. 
\end{lemma} 
\begin{proof} 
Let $q=w_{1}$. Since $|q|_{\diamond}\leq 1/2$, $w_{1}$ must be a digit. Since $s(q)\geq k|q|=k$, the claim follows. 
\end{proof} 

Let now $r$ be the largest index such that $w_{r}=k$. Let $s$ be the leftmost position $s>r$ such that $w_{s}=\diamond$. Let $t$ be the leftmost position $t>s$ such that $w_{t}\neq \diamond$, $t=2n+1$ if no such position exists. 

Consider the word $b$ obtained from $w$ by a. deleting positions $w_{r}$ and $w_{s}$. b. increasing the digit at position $w_{t}$ by one, if $t\neq 2n+1$. Note that, if $t\neq 2n+1$ then 
$w_{t}\neq k$, by the definition of index $r$. Also, $|b|=2n-2<2n$. 

$w$ is easily obtained from $b$ by inserting a $k$ in position $r$ and a diamond in position $s$, also deleting one lifeline 
from position $t$ if $t\neq 2n+1$. To complete the proof we need to argue that $b$ satisfies conditions 1-2a.b. 
Then, by induction, $b$ is an output of the process $HAD_{k,INT}$, hence so is $w$. 

Condition 1 is easy to check, since $|b|=2n-2$, and $b$ has exactly one $\diamond$ less than $w$, i.e. n-1 $\diamond$'s. As for 2.(a)-(b), let $p$ be a prefix of $b$. There are four cases: 

\begin{itemize} 
\item[-]\textbf{Case 1: $1\leq |p|< r$: } In this case $p$ is also a prefix of $w$, and the result follows from the inductive hypothesis. 
\item[-]\textbf{Case 2: $r\leq |p|< s-1$: } In this case $p=w_{1}\ldots w_{r-1}w_{r+1}\ldots w_{|p|+1}$. Let $z_{1}=w_{1}\ldots \ldots w_{|p|+1}$ be the corresponding prefix of $w$. 

The number of diamonds in $p$ is equal to the number of diamonds in $z_{1}$. Since $z_{1}$ does \textbf{not} end with a diamond (as $|p|<s-1$), the number of diamonds in $z_{1}$ is equal to that of its prefix $u$ of length $|p|$. By the 
induction hypothesis $|p|_{\diamond}=|u|_{\diamond}\leq |u|/2=|p|/2.$ So condition 2 (a). holds. On the other hand $s(p) + (k + 1)|p|_{\diamond} = (s(z_{1}) - k) + (k + 1)|z_{1}|_{\diamond} \geq k|z_{1}| -k = k|p|$, so 2 (b). holds as well. 

\item[-]\textbf{Case 3: $s-1\leq |p|< t-2$: } Thus $p=w_{1}\ldots w_{r-1}w_{r+1}\ldots w_{s-1}w_{s+1}\ldots w_{|p|+2}$. Let $z_{2}=w_{1}\ldots \ldots w_{|p|+2}$ be the corresponding prefix of $w$ of length $|p|+2$ and $z_{3}$ the prefix of $w$ of length $s-1$. 

The number of diamonds in $p$ is equal to the number of diamonds in $z_{2}$ minus one. By the induction hypothesis, 
this is at most $|z_{2}|/2-1$, which is at most $(|p|+2)/2-1=|p|/2$. Thus condition 2 (a). holds. Now $s(p)+ (k + 1)|p|_{\diamond} = $
\vspace{-2mm}
\begin{align*}
& (s(z_{2})-k) + (k + 1)(|z_{2}|_{\diamond} - 1) = (s(z_{3})-k)+(k+1)(|z_{3}|_{\diamond}+|z_{2}|-|z_{3}|-1) \\
 & \geq k|z_{3}|-k+(k+1)(|p|+2-|z_{3}|-1) = (k+1)(|p| + 1)-|z_{3}|-k= \\ & = (k+1)|p|-|z_{3}|+1  > k |p|+1+(|p|-|z_{3}|)> k|p|
\end{align*}

so condition 2 (b). is established as well. In the previous chain of (in)equalities we used the fact (valid by the very definition of $t$) that for all $s\leq i<t$, $w_{i}=\diamond.$

\item[-]\textbf{Case 4: $t-2\leq |p|\leq 2n$: }] In this case $p=w_{1}\ldots w_{r-1}w_{r+1}\ldots w_{s-1}w_{s+1}\ldots (w_{t}+1)\ldots $. Furthermore, $p$ ends with $w_{|p|+2}$ (if $|p|+2\neq t$) and with $w_{|p|+2}+1$  (if $|p|+2= t$). Let $z_{4}=w_{1}\ldots \ldots w_{|p|+2}$ be the prefix of $w$ of length $|p|+2$. 
\begin{itemize}
  \item[-] $|p|_{\diamond} = |z_{4}| - 1 \leq  |z_{4}|/ 2 - 1 = (|p|+2) / 2-1 = |p|/2.$
  \item[-] On the other hand $
 s(p) + (k + 1)|p|_{\diamond} = (s(z_{4})) - k + 1) + (k + 1)(|z_{4}|-1) \geq k|z_{4}| -k+1-k-1 = k \cdot (|p|+2)-2k=k|p|.$
\end{itemize}
so conditions 2 (a)-(b). are proved in this last case as well. 
\end{itemize} 
\end{proof} 

\section{Proof of Theorem~\ref{foo4}} 
\label{sec:pseries}

Justifying correctness of algorithm ComputeMultiplicity is simple: a string $w$ can result from any string $z$ by inserting a $k$ and deleting one life from the closest non-zero letter of $z$ to its right. After insertion, the new $k$ will be the rightmost element of a maximal block of $w$ of consecutive $k$'s. The letter it acts upon in $z$ cannot be a $k$ (in $w$), and cannot have any letters other than zero before it.

The candidates in $w$ for the changed letter are those letters $l$ succeeding the newly inserted $k$ such that $0\leq l\leq k-1$ and 
the only values between $k$ and $l$ are zeros. Thus these candidates are the following: (a)letters in $w$ forming the maximal block $B$ of zeros immediately following $k$ (if any), and (b).the first letter after $B$, provided it has value $0$ to $k-1$. Since we are counting multiplicities and all these words lead to distinct candidates, the correctness of the algorithm follows.

For $k=1$ the algorithm ComputeMultiplicity simplifies to a recurrence formula: Indeed, in this case there are no candidates of type (b). We derive:\\
$F_{1}([a_{1},b_{1},\ldots, a_{s},b_{s}])  = \sum_{\stackrel{i=1:}{ a_{i}> 1}}^{s} \sum_{\stackrel{j,l\geq 0}{j+1+l=b_{i}}} F_{1}([a_{1},\ldots, $ $a_{i}-1, j,1,l ,a_{i+1},\ldots,b_{s}]) + \sum_{\stackrel{i=1:}{ a_{i}= 1}}^{s} \sum_{\stackrel{j,l\geq 0}{j+1+l=b_{i}}} F_{1}([a_{1},\ldots, a_{i-1}, b_{i-1}+j,1,l ,a_{i+1}, $ $\ldots,  b_{s}])$ if $b_{s}>0,$ otherwise $
 F_{1}([a_{1},\ldots, a_{s},0]) = \sum_{\stackrel{i=1:}{ a_{i}> 1}}^{s-1} 
\sum_{\stackrel{j,l\geq 0}{j+1+l=b_{i}}} F_{1}([a_{1},$ $\ldots, a_{i}-1, j,1,l ,a_{i+1},\ldots,  a_{s},0]) + \sum_{\stackrel{i=1:}{ a_{i}= 1}}^{s-1} \sum_{\stackrel{j,l\geq 0}{j+1+l=b_{i}}} F_{1}([a_{1},\ldots, a_{i-1}, b_{i-1}+j,1,l, $ $a_{i+1},\ldots,  b_{s}])+ F_{1}([a_{1},\ldots,  a_{s}-1,0]).$

{\small 
\begin{figure}[ht]
\begin{minipage}{.50\textwidth}
\begin{tabular}{|c|c|c|c|c|c|}
\hline 
w  & 1 & 10 & 11 & 100 & 101   \\
\hline 
$F_{1}(w)$  &  1 & 1 & 1 & 1 & 2   \\ 
\hline 
w &  110 & 111 & 1000 & 1001 & 1010 \\ 
\hline 
$F_{1}(w)$ & 2 & 1 & 1 & 3 & 5 \\ 
\hline
w & 1011 & 1100 & 1101 & 1110 & 1111 \\ \hline
$F_{1}(w)$ & 3 & 5 & 3 & 3 &  1 \\
\hline 
\end{tabular}
\end{minipage}
\begin{minipage}{.45\textwidth}
 \begin{tabular}{|c|c|c|c|c|c|c|c|}
\hline 
w & 2 & 21 & 22 & 211 & 212 & 220 & 221 \\
\hline 
$F_{2}(w)$  & 1 & 1 & 1 & 1 & 2 & 1 & 1  \\ 
\hline 
w & 222 & 2111 & 2112 & 2120 & 2121 & 2122 & 2201  \\ 
\hline 
$F_{2}(w)$ & 1 & 1 & 3 & 2 & 3 & 3 & 1  \\ 
\hline 
w & 2202 & 2210 & 2211 & 2212 & 2220 & 2221 & 2222   \\ 
\hline 
$F_{2}(w)$ & 3 & 1 & 1 & 2 & 2  & 1 &  1   \\ 
\hline
\end{tabular}
\end{minipage}

\caption{The leading coefficients of formal power series (a). $F_1$. (b). $F_{2}$. }
\label{coeffs}
\end{figure}}

In spite of this, we weren't able to solve the recurrence above and compute the generating functions $F_{1}$ or, more generally, $F_{k}$, for $k\geq 1$.  An inspection of the coefficients obtained by the application of the algorithm is inconclusive: We tabulated the leading coefficients of series $F_{1}$ and $F_{2}$, computed using the Algorithm~\ref{algo} in Figures~\ref{coeffs} (a). and (b). The second listing is restricted to 2-dominant strings only.  No apparent closed-form formula for the coefficients of $F_{1},F_{2}$ emerges by inspecting these values.

\section{Application: estimating the value of the scaling constant $\lambda_{2}$.} 
\label{sec:apps} 
The computation of series $F_{2}$ allows us to tabulate (for small value of $n$) the values of the distribution of increments, a structural parameter whose limiting behavior determines the value of the constant $\lambda_{2}$ (conjectured, remember, to be equal to $\frac{1+\sqrt{5}}{2}$). 

\begin{definition} 
Let $w$ be a word that is an outcome of the process $HAD_{k}$. An {\em increment of $w$} is a position $p$ in $w$ (among the $|w|+1$ possible positions: at the beginning of $w$, at the end of $w$ or between two letters of $w$) such that no nonzero letters of $w$ appear to the right of $p$. The number of increments of word $w$ is denoted by $\#inc_{k}(w)$.  It is nothing but 1 plus the number of trailing zeros of $w$. 

 Let $L$ be an alphabet that contains $\Sigma_{k}$ for some $k\geq 1$. Given a word $w\in L^{*}$ we denote by $s(w)$ the sum of the digit characters of $w$. 
\end{definition}

The fact that increments  are useful in computing $\lambda_2$ is seen as follows: consider a word $w$ of length $n$ that is a sample from the $HAD_{k}$ process. Increments of $w$ are those positions where the insertion of a $k$ does not remove any lifeline, thus increasing the number of heaps in the corresponding greedy "patience heaping" algorithm \cite{istrate2015heapable} by 1. If $w$ has $t$ increment positions then the probability that the number of heaps will increase by one (given that the current state of the process is $w$) is $t/(n+1)$.

What we need to show is that (as $n\rightarrow \infty$) the mean number of positions that are increments in a random sample $w$ of length $n$ tends to $\lambda_{k}$. Therefore the probability that a new position will increase the number of heaps by 1 is asymptotically equal to $\lambda_{k}/(n+1)$. The scaling of the expected number of heaps follows from this limit.

In Figure~\ref{incs}. we plot the {\em exact} probability distribution of the number of increments (from which we subtract one, to make the distribution start from zero) for $k=2$ and several small values of $n$. They were computed exactly by employing Algorithm~\ref{algo} to exactly compute the probability of each string $w$, and then computing $\# inc_{2}(w)$. We performed this computation for $2\leq n\leq 13$. The corresponding expected values are tabulated (for all values $n=2,\ldots, 10$) in Figure~\ref{incs2}. 

\begin{figure}[ht]
\begin{center} 
\includegraphics[height=4cm, width=6cm]{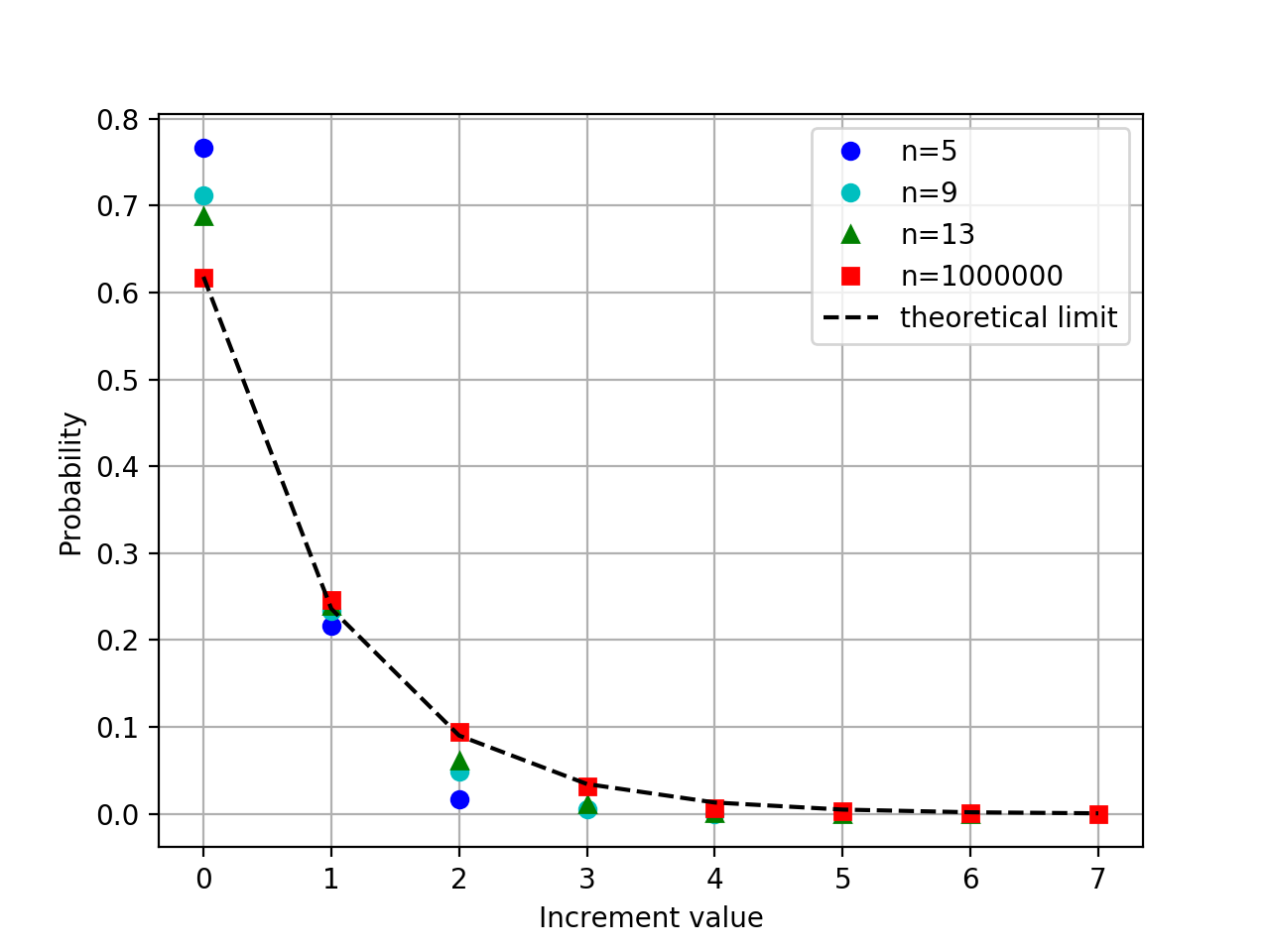}
\end{center} 
\caption{Probability distribution of increments, for $k=2$, and $n=5,9,13,1000000$.}
\label{incs}
\end{figure} 

\begin{figure}[ht]
\begin{center} 
\begin{tabular}{|c|c|c|c|c|c|c|c|}
\hline 
n & 2 & 3 & 4 & 5 & 6 & 7 \\ 
\hline 
$E[\#inc_{2}]$ & 1.0 & 1.166 & 1.208 & 1.250 & 1.281 & 1.307  \\ 
\hline 
n  & 8 & 9 & 10 & 100  & 100000 & 1 mil \\ 
\hline 
$E[\#inc_{2}]$& 1.329 & 1.347 & 1.363 & 1.520  & 1.575 & 1.580  \\ 
\hline 
\end{tabular} 
\end{center} 
\caption{The mean values of the distributions of increments.}
\label{incs2}
\end{figure} 

Unfortunately, as it turns out, the ability to exactly compute (for small values of $n$) the distribution of increments \textbf{does not} give an accurate estimate of the asymptotic behavior of this distribution, as the convergence seems rather slow, and not at all captured by these small values of $n$.  Indeed, to explore the distribution of increments for large values of $n$, as exact computation is no longer possible, we instead resorted to {\em sampling} from the distribution, by generating 10000 independent random trajectories of length $n$ from process $HAD_{2}$,  and then computing the distribution of increments of the sampled outcome strings. The outcome is presented  (for $n=100,100000,1000000$, together with some of the cases of the exact distribution) in Figure~\ref{incs}. The distribution of increments seems to converge (as $n\rightarrow \infty$) to a geometric distribution with parameter $p=\frac{\sqrt{5}-1}{2}\sim 0.618\cdots$. That is, we predict that for all $i\geq 1$, $ \lim_{n\rightarrow \infty}Pr_{|w|=n}[\# inc_{2}(w)=i]=p\cdot (1-p)^{i-1}.$
The fit between the (sampled) estimates for $n=1000000$ and the predicted limit distribution is quite good: every coefficient differs from its predicted value by no more than $0.003$, with the exception of the fourth coefficient, whose difference is $0.007$. Because of the formula for computing averages, these small differences have, though, a cumulative effect in the discrepancy for the average $E[\#inc_{2}(w)]$ for $n=10000000$ accounting for the $0.03$ difference between the sampled value and the predicted limit: in fact most of the difference is due to the fourth coefficient, as $4\times 0.007 = 0.028$. 
\begin{conclusion} 
The increment data supports the conjectured value $\lambda_2=1+p=\frac{1+\sqrt{5}}{2}.$
\end{conclusion} 
We intend to present (in the journal version of this paper) a similar investigation  of the value of constant $c_{k}$ in Conjecture~\ref{conj:int}.

\section{Open questions and future work}
\label{sec:open}

The major open problems raised by our work concerns the nature and asymptotic behavior of formal power series $F_{k}$, $F_{k,INT}$. An easy consequence of Corollary~\ref{complexity} is 

\begin{corollary} 
For $k\geq 2$ formal power series  $F_{k}$, $F_{k,INT}$ are {\bf not} {\bf N}-rational. 
\end{corollary} 

\begin{open} 
Are formal power series $F_{1}$, $F_{1,INT}$ {\bf N}-rational ? (We conjecture that the answer is negative).
\end{open}

Note that Reutenauer \cite{reutenauer1980series} extended the Chomsky-Sch\"utzenberger criterion for rationality from formal languages to power series: a formal power series is rational if and only if the so-called {\it syntactic algebra} associated to it has finite rank.  We don't know, though, how to explicitly apply this result to the formal power series we investigate in this paper.  On the other hand, in the general case, the characterization of context-free languages as supports of {\bf N}-algebraic series (e.g. Theorem 5 in ~\cite{petre2009algebraic} ), together with Theorem~\ref{noncfl}, establishes the fact that series $F_{k,INT}$ is not {\bf N}-algebraic. 

\begin{open} 
Are series $F_{k}$ {\bf N}-algebraic ? (Conjecture: the answer is negative).

\end{open}
\vspace{-0.5cm}

 
\nocite{*}
 {\scriptsize
 \bibliographystyle{abbrv}
 \bibliography{sample-dmtcs}
 
}
\end{document}